\documentclass[journal]{IEEEtran}
\usepackage{calc}

\usepackage{multirow}
\usepackage{cite}
\usepackage{graphicx,subfigure}
\usepackage{psfrag}
\usepackage{amsmath,amssymb}
\usepackage{color}
\usepackage{tikz}
\usepackage{pgfplots}
\usepackage{mathtools}

\usetikzlibrary{snakes}

\newenvironment{proof}{\begin{IEEEproof}}{\end{IEEEproof}}

\DeclareMathOperator*{\dotleq}{\overset{.}{\leq}}
\DeclareMathOperator*{\dotgeq}{\overset{.}{\geq}}

\newtheorem{theorem}{Theorem}
\newtheorem{corollary}{Corollary}[theorem]

\newtheorem{lemma}{Lemma}

\newtheorem{example}{Example} 

\newcommand{\bit}{\begin{itemize}}
\newcommand{\eit}{\end{itemize}}

\newcommand{\bc}{\begin{center}}
\newcommand{\ec}{\end{center}}

\newcommand{\ba}{\begin{array}}
\newcommand{\ea}{\end{array}}

\newcommand{\beq}{\begin{equation}}
\newcommand{\eeq}{\end{equation}}

\newcommand{\beqn}{\begin{equation*}}
\newcommand{\eeqn}{\end{equation*}}

\newcommand{\bean}{\begin{eqnarray*}}
\newcommand{\eean}{\end{eqnarray*}}
\newcommand{\bea}{\begin{eqnarray}}
\newcommand{\eea}{\end{eqnarray}}

\def\Z{\mathbb{Z}}

\def\C{\mathbb{C}}
\def\E{\mathbb{E}}

\def\gv{\boldsymbol{g}}
\def\hv{\boldsymbol{h}}

\def\xv{\boldsymbol{x}}

\begin{document}
\sloppy

\title{Feedback-Aided Coded Caching \\ for the MISO BC with Small Caches}
\author{Jingjing Zhang and Petros Elia
\thanks{The authors are with the Mobile Communications Department at EURECOM, Sophia Antipolis, 06410, France (email: jingjing.zhang@eurecom.fr, elia@eurecom.fr).
The work of Petros Elia was supported by the European Community's Seventh Framework Programme (FP7/2007-2013) / grant agreement no.318306 (NEWCOM\#), and from the ANR Jeunes Chercheurs project ECOLOGICAL-BITS-AND-FLOPS.}
\thanks{An initial version of this paper has been reported as Research Report No. RR-15-307 at EURECOM, August 25, 2015, http://www.eurecom.fr/publication/4723.}
}


\maketitle

\thispagestyle{empty}

\begin{abstract}
This work explores coded caching in the symmetric $K$-user cache-aided MISO BC with imperfect CSIT-type feedback, for the specific case where the cache size is much smaller than the library size. Building on the recently explored synergy between caching and delayed-CSIT, and building on the tradeoff between caching and CSIT quality, the work proposes new schemes that boost the impact of small caches, focusing on the case where the cumulative cache size is smaller than the library size.  For this small-cache setting, based on the proposed near-optimal schemes, the work identifies the optimal cache-aided degrees-of-freedom (DoF) performance within a factor of 4.
\end{abstract}

\section{Introduction\label{sec:intro}}

In the setting of the single-stream broadcast channel, the seminal work in \cite{MN14} proposed \emph{coded caching} as a technique which employed careful caching at the receivers, and proper coding across different users' requested data, to provide increased effective throughput and a reduced network load.
By using coding to create multicast opportunities --- even when users requested different data content --- coded caching allowed per-user DoF gains that were proportional to the cache sizes. The fact though that such caches can be comparably small \cite{EJR:15}, brings to the fore the need to understand how to efficiently combine reduced caching resources, with any additional complementary resources --- such as feedback and spatial dimensions --- that may be available in communication networks.

Our aim here is to explore this concept, in the symmetric $K$-user cache-aided wireless MISO BC. Following in the footsteps of \cite{ZEinterplay:16}, our aim here is to further our understanding of the effect of coded caching  --- now with small caches --- and (variable quality) feedback, in \emph{jointly} removing interference and improving performance. This joint exposition is natural and important because caching and feedback are both powerful and scarce ingredients in wireless networks, and because these two ingredients are intimately connected. These connections will prove particularly crucial here, in boosting the effect of otherwise insufficiently large caches, or otherwise insufficiently refined feedback. The coding challenge here --- for the particular case of small caches --- is to find a way to ameliorate the negative effect of having to leave some content entirely uncached, which is a problem which paradoxically can become more pronounced in the presence of CSIT resources, as we will see later on.

\subsection{$K$-user feedback-aided symmetric MISO BC with small caches}
We consider the symmetric $K$-user wireless MISO BC with a $K$-antenna transmitter, and $K$ single-antenna receiving users. The transmitter has access to a library of $N\geq K$ distinct files $W_1,W_2, \dots, W_N$, each of size $|W_n| = f$ bits. Each user $k \in \{1,2,\dots,K\}$ has a cache $Z_k$, of size $|Z_k| = Mf$ bits, where naturally $M \leq N$. A normalized measure of caching resources will take the form
\begin{align} \label{eq:gamma1}
\gamma := \frac{M}{N}.
\end{align}
Our emphasis here will be on the small cache regime where the cumulative cache size is less than the library size ($K\gamma\leq 1$, i.e., $KM\leq N$), and which will force us to account for the fact that not all content can be cached. We will also touch upon the general small-cache setting where the individual cache size is much less than the library size ($M\ll N$, i.e., $\gamma\ll 1$).

As in \cite{MN14}, communication consists of the aforementioned \emph{content placement phase} (typically taking place during off-peak hours) and the \emph{delivery phase}.
During the placement phase, the caches are pre-filled with content from the $N$ files $\{W_n\}_{n=1}^{N}$ of the library. The delivery phase commences when each user $k=1,\dots,K$ requests from the transmitter, any \emph{one} file $W_{R_k}\in \{W_n\}_{n=1}^{N}$, out of the $N$ library files. Upon notification of the users' requests, the transmitter aims to deliver the (remaining of the) requested files, each to 
their intended receiver, and the challenge is to do so over a limited (delivery phase) duration $T$.

\paragraph{Channel model}
For each transmission, the received signals at each user $k$, will be modeled as
\begin{align}
y_{k}=\hv_{k}^{T} \xv + z_{k}, ~~ k = 1, \dots, K
\end{align}
where $\xv\in\mathbb{C}^{K\times 1}$ denotes the transmitted vector satisfying a power constraint $\E(||\xv||^2)\leq P$, where $\hv_{k}\in\mathbb{C}^{K\times 1}$ denotes the channel of user $k$ in the form of the random vector of fading coefficients that can change in time and space, and where $z_{k}$ represents unit-power AWGN noise at receiver $k$.
\begin{figure}[t!]
  \centering
\includegraphics[width=0.8\columnwidth]{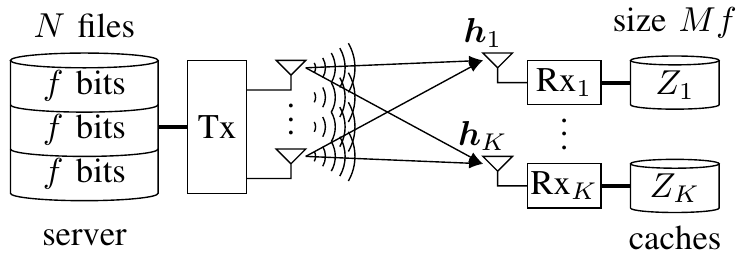}
\caption{Cache-aided $K$-user MISO BC.}
\label{fig:model}
\end{figure}
At the end of the delivery phase, each receiving user $k$ combines the received signal observations $y_{k}$ --- accumulated during the delivery phase --- with the fixed information in their respective cache $Z_k$, to reconstruct their desired file $W_{R_k}$.

\paragraph{Feedback model}
Communication will also take place in the presence of channel state information at the transmitter. Motivated by the fact that CSIT-type feedback is typically hard to obtain in a timely and reliable manner, we will here consider the mixed CSIT model (cf.~\cite{YKGY:12d}, see also \cite{CE:13it}) where feedback offers a combination of imperfect-quality current (instantaneously available) CSIT, together with additional (perfect-accuracy) delayed CSIT. In this setting, the channel estimation error of the current channel state is assumed to scale in power as $P^{-\alpha}$, for some CSIT quality exponent\footnote{The range of interest is $\alpha\in[0,1]$; in the high SNR regime of interest here, $\alpha=0$ corresponds to having essentially no current CSIT, while having $\alpha = 1$ corresponds (again in the high SNR regime) to perfect and immediately available CSIT (cf.~\cite{Caire+:10m}).}

\begin{align}
\alpha & := -\lim_{P \rightarrow \infty} \frac{\log \E[||{\hv_{k}}-{\hat \hv_{k}}||^2]}{\log P}, ~k\in \{1,\dots,K\}
\end{align}
where ${\hv_{k}}-{\hat \hv_{k}}$ denotes the estimation error between the current CSIT estimate ${\hat \hv_{k}}$ and the estimated channel ${\hv_{k}}$.

This mixed CSIT model, in addition to being able to capture different realistic scenarios such as that of using predictions and feedback to get an estimate of the current state of a time-correlated channel, it is also well suited for cache aided networks because, by mixing the effect of delayed and current feedback, we can concisely capture the powerful synergies between caching and delayed CSIT (cf.~\cite{ZEsynergy:16}) as well as the tradeoffs between the necessary feedback quality and cache size (cf.~\cite{ZEinterplay:16}).

\subsection{Measures of performance, notation and assumptions}

\subsubsection{Measures of performance in current work}
Our aim is to design schemes that, for any $K\gamma<1$ and any $\alpha\in[0,1]$, reduce the duration $T(\gamma,\alpha)$ --- in time slots, per file served per user --- needed to complete the delivery process, \emph{for any request}.
Equivalently, when meaningful, we will also consider the \emph{cache-aided degrees of freedom per user} (cache-aided DoF) which is simply\footnote{We note that $K d(\gamma,\alpha)$ is simply the coding gain $K(1-\gamma)/T$ that is often used to quantify the gain from coded-caching.}
\begin{align}  \label{eq:TtoDoF}
d(\gamma,\alpha)=\frac{1-\gamma}{T} \in [0,1].
\end{align}
\subsubsection{Notation}
We will use
\begin{align}
\Gamma := \frac{KM}{N} = K\gamma
\end{align}
to represent the cumulative (normalized) cache size, in the sense that the sum of the sizes of the caches across all users, is a fraction $\Gamma$ of the volume of the $N$-file library. We will also use the notation $H_n := \sum_{i=1}^{n} \frac{1}{i}$,
to represent the $n$-th harmonic number, and we will use $\epsilon_n := H_n-\log (n)$ to represent its logarithmic approximation error, for some integer $n$. We remind the reader that $\epsilon_n$ decreases with $n$, and that $\epsilon_\infty :=\lim \limits_{n \rightarrow \infty} H_n -  \log (n) $ is approximately $0.5772$.
$\mathbb{Z}$ will represent the integers, $\mathbb{Z}^{+}$ the positive integers, $\mathbb{R}$ the real numbers, $\binom{n}{k}$ the $n$-choose-$k$ operator, and $\oplus$ the bitwise XOR operation. We will use $[K]:= \{1,2,\cdots,K\}$. If $\psi$ is a set, then $|\psi|$ will denote its cardinality. For sets $A$ and $B$, then $A \backslash B$ denotes the difference set.
Complex vectors will be denoted by lower-case bold font. We will use $||\xv||^2$ to denote the magnitude of a vector $\xv$ of complex numbers. For a transmitted vector $\xv$, we will use $\text{dur}(\xv)$ to denote the transmission duration of that vector. For example, having $\text{dur}(\xv) = \frac{1}{10}T$ would simply mean that the transmission of vector $\xv$ lasts one tenth of the delivery phase.
We will also use $\doteq$ to denote \emph{exponential equality}, i.e., we write $g(P)\doteq P^{B}$ to denote $\displaystyle\lim_{P\to\infty}\frac{\log g(P)}{\log P}=B$.  Similarly $\dotgeq$ and $\dotleq$ will denote exponential inequalities.  Logarithms are of base~$e$, unless we use $\log_2(\cdot)$ which will represent a logarithm of base~2.

\subsubsection{Main assumptions}
In addition to the aforementioned mixed CSIT assumptions, we will adhere to the common convention (see for example~\cite{MAT:11c}) of assuming perfect and global knowledge of delayed channel state information at the receivers (delayed global CSIR), where each receiver must know (with delay) the CSIR of (some of the) other receivers. We will assume that the entries of \emph{each specific} estimation error vector are i.i.d. Gaussian.
For the outer (lower) bound to hold, we will make the common assumption that the current channel state must be independent of the previous channel-estimates and estimation errors, \emph{conditioned on the current estimate} (there is no need for the channel to be i.i.d. in time). Furthermore, as with most works on coded caching, we will assume uniform file popularity, as well as that $N\geq K$.

\subsection{Prior work}
In terms of feedback, our work builds on many different works including~\cite{MAT:11c}, as well as other subsequent works~\cite{YKGY:12d,CE:13it,GJ:12o,CE:12d,KYG:13,CYE:13isit,VV:09,TJSP:12,LH:12,HC:13,CJS:07} that incorporate different CSIT-quality considerations. 
In terms of caching, our work is part of a sequence of works (cf.~\cite{WLTL:15,MND13,JTLC:14,HA:2015,WLG:15,APPV:15}) that is inspired by the work in~\cite{MN14} and which try to understand the limits of coded caching in different scenarios. Additional interesting works include \cite{GSDMC:12,PBKD:15,JWTLCEL:15,NSW:12,BBD:15,MCOFBJ:14,HKD:14,HKS:15,SJTLD:15,JTLC:14,DBAD:15}, as well as the work in~\cite{CFLsmallCaches:14} which considered coded caching --- in the single stream case with $K\geq N$ --- in the presence of very small caches with $KM \leq 1$, corresponding to the case where pooling all the caches together, can at most match the size of a single file.

In spirit, our work is closer to different works that deviate from the setting of having single-stream error free links, such as the works by Timo and Wigger in~\cite{TW:15} and by Ghorbel et al. \cite{GKY:15} on the cache-aided erasure broadcast channel, the work by Maddah-Ali and Niesen in~\cite{MN:15isit} on the wireless interference channel with transmitter-side caching, and our work in~\cite{ZFE:15}.

\section{Main results\label{sec:mainResults}}

The following identifies, up to a factor of 4, the optimal $T^*$, for all $\Gamma \in [0,1]$. We use the expression
\begin{align} \label{eq:alphaBreak}
\alpha_{b,\eta} = \frac{\eta-\Gamma}{\Gamma(H_K-H_\eta-1)+\eta},  \ \eta = 1,\dots,K-1. \end{align}

\begin{theorem}\label{thm:smallGamma}
In the $(K,M,N,\alpha)$ cache-aided MISO BC with $N$ files, $K\leq N$ users, and $KM \leq N$ ($\Gamma \leq 1$),
then for $\eta = 1,\dots,K-2$,
\begin{align}
T =\left\{ {\begin{array}{*{20}{c}}
\frac{H_K-\Gamma}{1-\alpha+\alpha H_K}, & 0 \leq \alpha < \alpha_{b,1}\\
\frac{(K-\Gamma)(H_K-H_\eta)}{(K-\eta)+\alpha(\eta+K(H_K-H_\eta-1))}, & \alpha_{b,\eta} \leq \alpha < \alpha_{b,\eta+1} \\
1-\gamma, &\frac{K-1-\Gamma}{(K-1)(1-\gamma)} \leq \alpha \leq 1
\end{array}} \right.   \label{eq:gammasmall}
\end{align}
is achievable, and has a gap from optimal that is less than 4 ($\frac{T}{T^*}<4$), for all $\alpha,K$. For $\alpha \geq \frac{K-1-\Gamma}{(K-1)(1-\gamma)} $, $T$ is optimal.
\end{theorem}
\vspace{3pt}
\begin{proof}
The scheme that achieves the above performance is presented in Section~\ref{sec:schemeAlphaGammaSmall}, while the corresponding gap to optimal is bounded in Section~\ref{sec:gapCalculation}.
\end{proof}

\vspace{3pt}

Furthermore directly from the above, for $\alpha = 0$, we have the following.

\vspace{3pt}
\begin{corollary} \label{cor:noCSITsmallGamma}
In the MISO BC with $\Gamma \leq 1,\alpha = 0$, then
\begin{align}
T = H_{K}-\Gamma \label{eq:noagammasmall}
\end{align}
is achievable and has a gap from optimal that is less than 4.
\end{corollary}

\vspace{3pt}
Directly from Theorem~\ref{thm:smallGamma}, we have the following corollary which quantifies the CSIT savings
\begin{align}
\label{eq:alphaGainCode}
\delta(\gamma,\alpha)\! := \!\arg\min_{\alpha'}\{\alpha': \!(1-\gamma) T^*(\gamma=0,\alpha') \!\leq \!T(\gamma,\alpha)\}-\!\alpha \nonumber
\end{align}
that we can have as a result of properly exploiting small caches. This reflects the CSIT reductions (from $\alpha+\delta(\gamma,\alpha)$ to the operational $\alpha$) that can be achieved due to coded caching, without loss in performance.
\vspace{3pt}
\begin{corollary} \label{cor:AlphaGain_smallGamma}
In the $(K,M,N,\alpha)$ cache-aided BC with $\Gamma \leq  1$, then
\begin{align}
\delta(\gamma,\alpha) =\left\{ {\begin{array}{*{20}{c}}
\frac{\gamma(K- H_K)}{H_K - K \gamma}(\alpha + \frac{1}{H_K-1}) , & 0 \leq \alpha < \alpha_{b,1}\\
 \frac{(1-\alpha)(KH_\eta-\eta H_K)}{KH_{\eta+1}(H_K-1)}, & \!\!\!  \! \!  \!\!\! \alpha_{b,\eta} \leq \alpha < \alpha_{b,\eta+1} \\
1-\alpha, &  \!\!\!  \! \alpha\geq \frac{K(1-\gamma)-1}{(K-1)(1-\gamma)}.
\end{array}} \right.
\end{align}
\end{corollary}
\vspace{3pt}
The last case in the above equation shows how, in the presence of caching, we need not acquire CSIT quality that exceeds $\alpha=\frac{K(1-\gamma)-1}{(K-1)(1-\gamma)}$.

\paragraph{Tightening the bounds for the large BC with scalably small caches ($K\gg 1$, \ $\gamma\ll 1$)}
We now briefly touch upon the more general small-cache setting of $\gamma\ll 1$, where we have a large number of users $K\gg 1$. In this setting --- which captures our case of $\Gamma\leq 1$, as well as the case where $\Gamma> 1$ but where still $\gamma\ll 1$ --- we tighten the gap to optimal for the achievable performance, here (from Theorem~\ref{thm:smallGamma}), as well as for the $\Gamma\geq 1$ setting in \cite{ZEinterplay:16} which stated that
\begin{align}
\label{eq:GammaLarge}
T_{\Gamma\geq 1} := \frac{(1-\gamma)(H_K-H_{\Gamma})}{\alpha(H_K-H_{\Gamma})+(1-\alpha)(1-\gamma)}.
\end{align}
\vspace{3pt}
\begin{theorem}\label{thm:asymptotic}
In the $(K,M,N,\alpha)$ cache-aided MISO BC, in the limit of large $K$ and reduced cache size $M\ll N$, the achieved $T$ from Theorem~\ref{thm:smallGamma} (as well as $T_{\Gamma\geq 1}$), are at most a factor of 2 from optimal, for all values of $\alpha$.
\end{theorem}
\vspace{3pt}
\begin{proof}
The proof is found in Appendix~\ref{sec:asymptoticProofSmallGamma}. 
\end{proof}

The following shows (for the case of $\alpha = 0$) how, even a vanishing $\gamma = \frac{M}{N}\rightarrow 0$, can provide a non-vanishing gain\footnote{To avoid confusion, we clarify that the main Theorem is simply a DoF-type result, where SNR scales to infinity, and where the derived DoF holds for all $K$. The corollary below is simply based on the original DoF expression (i.e., SNR diverges first), which is then approximated in the large $K$ setting ($K$ diverges second, simultaneously with $\gamma$).}.
\vspace{3pt}
\begin{corollary}\label{cor:asymptotic2}
In the $(K,M,N,\alpha=0)$ cache-aided MISO BC, as $K$ scales to infinity and as $\gamma$ scales as $\gamma = K^{-(1-\zeta)}$ for any $\zeta \in[0,1)$, the gain from caching is
\begin{align}
\lim_{K\rightarrow \infty} \frac{T(\gamma = K^{-(1-\zeta)},\alpha=0)}{T^{*}(\gamma = 0,\alpha = 0)} =1-\zeta.
\end{align}
\end{corollary}
\vspace{3pt}
\begin{proof}
The expression follows directly from \eqref{eq:GammaLarge}.\end{proof}
\vspace{3pt}

\begin{example}
Consider a future large MIMO system with $K = N = 1000$. In the absence of caching, the optimal performance is $T^*(\gamma = 0, \alpha = 0)\approx log (K)\approx 6.91$ (cf.~\cite{MAT:11c} and \eqref{eq:TtoDoF}). Assume now that we introduce modest (coded) caching with $M = K^\zeta = \sqrt{K} \approx 31.6$ ($\zeta = 1/2$, $\gamma = \frac{31.6}{1000}\approx 0.03$), then the optimal reduction --- due to caching --- is described in Corollary~\ref{cor:asymptotic2} to approach a multiplicative factor of $1-\zeta = \frac{1}{2}$, corresponding to a reduction from $T^*(\gamma = 0, \alpha = 0)\approx 6.91$ to about half of that (doubling the DoF).  On the other hand, for the same $M,N$, if only local caching (data push) techniques were used, without coded caching, the best caching gain\footnote{Optimality is direct after using basic cut-set bound techniques, which can tell us that there must exist a $k\in\{1,\dots,K\}$ such that $|W_{R_k} \backslash Z_k|\geq (1-\gamma)f$.} would take the form $\frac{T'}{T^*(\gamma = 0, \alpha = 0)}  = 1-\gamma = 1-K^{-1/2} \approx 0.97$ which corresponds to a reduction from $T^*(\gamma = 0, \alpha = 0)$ by only about 3 \%.
\end{example}

\section{Cache-aided QMAT with very small caches\label{sec:schemeAlphaGammaSmall}} 

We now describe the communication scheme. Part of the challenge, and a notable difference from the case of larger caches, is that due to the fact that now $\Gamma<1$, some of the library content must remain entirely uncached. This uncached part is delivered by employing a combination of multicasting and ZF which uses current CSIT. The problem though remains when $\alpha$ is small because then current CSIT can only support a weak ZF component, which will in turn force us to send some of this uncached private information using multicasting, which itself will be calibrated not to intervene with the multicasting that utilizes side information from the caches. For this range of smaller $\alpha$, our scheme here will differ from that when $\alpha$ is big (as well as from the scheme for $\Gamma\geq 1$). When $\alpha$ is bigger than a certain threshold value $\alpha_{b,1}$, we will choose to cache even less data from the library, which though we will cache with higher redundancy\footnote{Higher redundancy here implies that parts of files will be replicated in more caches.}. Calibrating this redundancy as a function of $\alpha$, will allow us to strike the proper balance between ZF and delayed-CSIT aided coded caching. For this latter part, we will use our scheme from \cite{ZEinterplay:16} which we do not describe here.

We consider the range\footnote{The case of $\alpha \geq \alpha_{b,1}$ will be briefly addressed at the end of this section.} $\alpha \in [0, \alpha_{b,1}]$, and proceed to set $\eta = 1$ (cf.~\eqref{eq:alphaBreak} from Theorem~\ref{thm:smallGamma}), such that there is no overlapping content in the caches ($Z_k\cap Z_i =\emptyset$).
\subsection{Placement phase}
During the placement phase, each of the $N$ files $W_n, n = 1, 2, \ldots, N$ ($|W_n| = f$ bits) in the library, is split into two parts
\begin{align} \label{eq:splitCachedUncached}
W_n = (W_n^c, W_n^{\overline{c}})
\end{align}
where $W_n^c$ ($c$ for `cached') will be placed in different caches, while the content of $W_n^{\overline{c}}$ ($\overline{c}$ for `non-cached') will never be cached anywhere, but will instead be communicated --- using current and delayed CSIT --- in a manner that avoids interference without depending on caches.
The split is such that
\begin{align}\label{eq:WNcSize}
|W_n^c| = \frac{KMf}{N} = K \gamma f  \ \text{bits}.
\end{align}
Then, we equally divide $W_n^c$ into $K$ subfiles $\{W^c_{n,k}\}_{k \in [K]} $, where each subfile has size
\begin{align} \label{eq:WnTauSize}
|W^c_{n,k}| = \frac{Mf}{N} =\gamma f \ \text{bits}
\end{align}
and the caches are filled as follows
\begin{align}\label{eq:ZkFill1} Z_k=\{W^c_{n,k}\}_{n \in [N]}\end{align}
such that each subfile $W^c_{n, k}$ is stored in $Z_k$.

\subsection{Delivery phase}
Upon notification of the requests $W_{R_k}, k=1,\dots,K$, we first further split $W^{\overline{c}}_{R_k,k}$
into two parts, $W^{\overline{c},p}_{R_k,k}$ and $W^{\overline{c},\overline{p}}_{R_k,k}$ that will be delivered in two different ways that we describe later, and whose sizes are such that
\begin{align} \label{eq:WRktauSplitSizesSmallGamma}
|W^{\overline{c},p}_{R_k,k}| = \alpha f T, \ \ \
|W^{\overline{c},\overline{p}}_{R_k,k}| = f(1-K \gamma-\alpha T).
\end{align}

Then we fold all $W^{c}_{R_k,\psi \backslash \{k\}}$ to get a set
\begin{align} \label{eq:XpsiDefSmallGamma}
X_{\psi} := \oplus_{k \in \psi} W^{c}_{R_k,\psi \backslash \{k\}}, \psi \in  \Psi_{2}
\end{align} of so-called \emph{order-2 XORs} (each XOR is meant for two users), and where $\Psi_{2} := \{\psi\in [K] \ : \  |\psi|=2  \}$.
Each of these XORs has size
\begin{align} \label{eq:XpsiSizeSmallGamma}
|X_{\psi}|  = \gamma f  \ \text{bits}
\end{align}
and they jointly form the XOR set
\begin{align}\label{eq:foldedMessagesSmallGamma}
 \mathcal{X}_\Psi := \{ X_{\psi} = \oplus_{k \in \psi} W^{c}_{R_k,\psi \backslash \{k\}}\}_{\psi \in  \Psi_{2}}\end{align}
of cardinality $|\mathcal{X}_\Psi|=\binom{K}{2}$.

In the end, we must deliver
\bit
\item $W_{R_k}^{\overline{c},p}, \ k =1,\cdots,K$, privately to user $k$, using mainly current CSIT
\item $W_{R_k}^{\overline{c},\overline{p}}, \ k = 1,\cdots,K$, using mainly delayed CSIT
\item $\{W^{c}_{R_k,\psi \backslash \{k\}}\}_{\psi \in  \Psi_{2}}$, $k = 1,\cdots,K$ by delivering the XORs from $\mathcal{X}_\Psi$, each to their intended pair of receivers.
\eit
This delivery is described in the following.


\paragraph{Transmission}
We describe how we adapt the QMAT algorithm from~\cite{KGZE:16} to deliver the aforementioned messages, with delay $T$.

While we will not go into all the details of the QMAT scheme, we note that some aspects of this scheme are similar to MAT (cf.\cite{MAT:11c}), and a main new element is that QMAT applies digital transmission of interference, and a double-quantization method that collects and distributes residual interference across different rounds, in a manner that allows for ZF and MAT to coexist at maximal rates. The main ingredients include MAT-type symbols of different degrees of multicasting, ZF-type symbols for each user, and auxiliary symbols that diffuse interference across different phases and rounds. Many of the details of this scheme are `hidden' behind the choice of $\textbf{G}_{c,t}$ and behind the loading of the MAT-type symbols and additional auxiliary symbols that are all represented by $\xv_{c,t}$ below.  Another important element involves the use of caches to `skip' MAT phases, as well as a careful rate- and power-allocation policy.

The QMAT algorithm has $K$ transmission phases. For each phase $ i=1,\cdots,K$, the QMAT data symbols are intended for a subset $\mathcal{S} \subset [K]$ of users, where $|\mathcal{S}|=i$. Here by adapting the algorithm, at each instance $t\in[0, T]$ throughout the delivery phase, the transmitted vector takes the form
\begin{align} \label{txformperfect}
\xv_{t} = \textbf{G}_{c,t} \xv_{c,t}+ \sum_{\ell\in \bar{\mathcal{S}}}\gv_{\ell,t} a_{\ell,t}^{*}  +\sum_{k=1}^{K} \gv_{k,t} a_{k,t}
\end{align}
with $\xv_{c,t}$ being a $K$-length vector for QMAT data symbols, with $a_{\ell,t}^{*}$ being an auxiliary symbol that carries residual interference, where $\bar{\mathcal{S}}$ is a set of `undesired' users that changes every phase, and where each unit-norm precoder $\gv_{k,t}$ for user $k=1,2,\dots,K$, is simultaneously orthogonal to the CSI estimate for the channels of all other users ($\gv_{l,t}$ acts the same), thus guaranteeing
\begin{align}
\hat{\hv}_{k',t}^{T} \gv_{k,t} = 0, \ \ \forall k' \in [K] \backslash k.
\end{align}
Each precoder $\textbf{G}_{c,t}$ is defined as $\textbf{G}_{c,t} = [\gv_{c,t}, \textbf{U}_{c,t}]$, where $\gv_{c,t}$ is simultaneously orthogonal to the channel estimates of the undesired receivers, and $\textbf{U}_{c,t} \in \C^{K\times(K-1)}$ is a randomly chosen, isotropically distributed unitary matrix.

The rates and the power are set by the QMAT algorithm, such that:
\bit
\item each $\xv_{c,t}$ and $a_{\ell,t}^{*}$ carries $f(1-\alpha)\text{dur}(\xv_{c,t}) $ bits,
\item each $a_{\ell,t}^{*}$ carries $\min\{ f(1-\alpha), f \alpha\} \text{dur}(\gv_{\ell,t} a_{\ell,t}^{*}) $ bits,
\item each $a_{k,t}$ carries $f \alpha \text{dur}(\gv_{k,t}a_{k,t}) $ bits,
\item and 
\begin{align}
\E\{|\xv_{c,t}|_1^2\} &=  \E\{|a_{\ell,t}^{*}|^2\} \doteq  P \notag \\  \E\{|\xv_{c,t}|_i^2\} &= \E\{|a_{k,t}|^2\} \doteq P^{\alpha} \notag
\end{align}
where $|\xv_{c,t}|_i, i=1,2,\cdots,K,$ denotes the magnitude of the $i^{th}$ entry of vector $\xv_{c,t}$.
\eit
The scheme here employs a total of $2K-1$ phases (rather than the $K$ phases in the original Q-MAT), where during the first $K-1$ phases (labeled here as phases $j=1,\dots,K-1$), the vector $\xv_{c,t}$ carries the folded messages $X_\psi \in \mathcal{X}_\Psi$ using the last $K-1$ phases of the MAT algorithm from~\cite{MAT:11c}, while for phases $j=K,\dots,2K-1$, $\xv_{c,t}$ now carries $\{W_{R_k}^{\overline{c},\overline{p}}\}_{k \in [K]}$ using the entirety of MAT. In addition, for all $2K-1$ phases, the different $a_{k,t}$ will carry (via ZF) all of the uncached $W^{\overline{c},p}_{R_k}, k=1,\dots,K$. The power and rate allocation guarantee that these MAT and ZF components can be carried out in parallel with the assistance of the auxiliary symbols from the next round\footnote{We here focus, for ease of description, on describing only one round. For more details on the multi-round structure of the QMAT, please see~\cite{KGZE:16}.}.

In the following, we use $T_j$ to denote the duration of phase $j$, and $T^{(1)} := \sum_{j=1}^{K-1}T_j$ to denote the duration of the first $K-1$ phases.

\paragraph{Summary of the transmission scheme for delivery of $\{X_{\psi}\}_{\psi \in  \Psi_{2}}$} 
Here, $\xv_{c,t}, \ t\in[0,T^{(1)}]$ will have the structure defined by the last $K-1$ phases of (one round of) the QMAT algorithm.

During the first phase ($t \in [0, T_{1} ]$, corresponding to phase $2$ of QMAT, where $|\mathcal{S}|=2$), $\xv_{c,t}$ will convey all the order-2 messages in $\{X_{\psi}\}_{\psi \in  \Psi_{2}}$ (each $\psi$ corresponds to each $\mathcal{S}$). Then, at the end of this phase, for each $\psi \in \Psi_{2}$, and for each $k\in \psi$, the received signal at user $k$, takes the form
\begin{align}
y_{k,t} = \underbrace{\hv_{k,t}^{T} \textbf{G}_{c,t} \xv_{c,t} }_{  \ \text{power} \ \doteq \ P}+ \underbrace{\hv_{k,t}^{T} \sum_{\ell \in \psi}\gv_{\ell,t}  a_{\ell,t} }_{\doteq \ P^{1-\alpha}}+  \underbrace{\hv_{k,t}^{T} \gv_{k,t} a_{k,t}}_{P^{\alpha}}
\end{align}
while the received signal for the other users $k \in [K]\backslash \psi$ takes the form
\begin{align}
y_{k,t} =  \underbrace{\hv_{k,t}^{T} \gv_{k,t}  a_{k,t}^{*} }_{\ \text{power} \ \doteq \ P} \!  + \! \underbrace{\hv_{k,t}^{T} (\sum_{\substack{\ell \in \psi  \\ \ell \neq k}}\gv_{\ell,t}  a_{\ell,t}^{*} \!  + \!  \textbf{G}_{c,t} \xv_{c,t})}_{L_{\psi,k'},  \ \ \doteq \ P^{1-\alpha}}\!  + \!  \underbrace{\hv_{k,t}^{T} \gv_{k,t} a_{k,t}}_{P^{\alpha}}
\end{align}
where in both cases, we ignored the Gaussian noise and the ZF noise up to $P^{0}$.
Following basic MAT techniques, the interference $L_{\psi,k'}, \forall k'$ is translated into order-3 messages and will be sent in phase $j=2$. In addition, to separate $\hv_{k,t}^{T} \gv_{k,t} a_{k,t}$ from the MAT component, as in~\cite{KGZE:16}, we use auxiliary data symbols $a_{k,t}^{*}$. Specifically, $L_{\psi,k'}$ is first quantized with $(1-2\alpha)^+ \log P$ bits, leaving the quantization noise $n_{\psi,k'}$ with power scaling in $P^{\alpha}$. Then, the transmitter quantizes this quantization noise $n_{\psi,k'}$ with $\alpha \log P$ bits up to the noise level, which will be carried by the auxiliary data symbols in the corresponding phase in the next round. In this way, $\xv_{c,t}$ \emph{can be decoded using the auxiliary data symbols of the next round, and using order-3 messages from the next phase}.

Given that the allocated `rate' for $\xv_{c,t}$ is $(1-\alpha)f$, and given that there is a total of $|\mathcal{X}_\Psi|=\binom{K}{2}$ different order-2 folded messages $X_\psi$ ($|X_\psi| = \gamma f$ bits), the duration $T_1$ of the first phase, takes the form $
T_1 = \frac{ \binom{K}{2}  |X_{\psi}|}{(K-1)(1-\alpha)f} = \frac{ \gamma \binom{K}{2}}{(K-1)(1-\alpha)}.$

For phases $j=2,\cdots,K-1$ here (which draw from the last $K-2$ phases in \cite{KGZE:16}), we can similarly calculate their duration to be $T_j = \frac{2}{j+1} T_1$, which in turn implies that
\begin{align}
T^{(1)} = \sum \limits_{j=1}^{K-1} T_j = T_{1} \sum \limits_{j=1}^{K-1} \frac{2}{1+j}  = \frac{\Gamma(H_K-1)}{1-\alpha}. \label{eq:durationpart1}
\end{align}

\paragraph{Transmission of $\{W^{\overline{c},\overline{p}}_{R_k}\}_{k\in [K]}$}
Now the remaining information from $\{W_{R_k}^{\overline{c},\overline{p}}\}_{k \in [K]}$, will be conveyed by $\xv_{c,t}, t \in [T^{(1)}, T]$ (phases $K,\dots,2K-1$), where now though we will use all the phases of the Q-MAT algorithm because now there is no corresponding side information in the caches to help us `skip' phases. During the first phase of this second part (i.e., during phase $j=K$), we place all of $\{W_{R_k}^{\overline{c},\overline{p}}\}_{k \in [K]}$ in $\xv_{c,t}, t \in [T^{(1)}, T^{(1)}+T_K]$. 
Given the allocated rate $(1-\alpha)f$ for $\xv_{c,t}$, and given that $|\{W_{R_k}^{\overline{c},\overline{p}}\}_{k \in [K]}|= Kf(1-\frac{KM}{N}-\alpha T)$, we see that
\begin{align}
T_K = \frac{Kf(1-\frac{KM}{N}-\alpha T)}{K(1-\alpha)f}  = \frac{(1-\Gamma-\alpha T)}{(1-\alpha)}.
\end{align}
Similarly we see that $T_j = \frac{1}{j-K+1} T_K,  \ j=K,\dots,2K-1$, which means that
\begin{align}
T-T^{(1)} &= \sum \limits_{j=K}^{2K-1} T_j = T_{K} \sum \limits_{j=K}^{2K-1} \frac{1}{j-K+1} \notag \\
 &= \frac{H_K(1-\Gamma-\alpha T)}{(1-\alpha)}. \label{eq:durationpart2}
\end{align}
Combining~\eqref{eq:durationpart1}~and~\eqref{eq:durationpart2}, gives the desired
\begin{align}
T = \frac{H_K-  \Gamma}{1-\alpha+\alpha H_K}.
\end{align}

\paragraph{Communication scheme for $\alpha \in [\alpha_{b,1}, \alpha_{b,K-1}]$ }
Here, when $\alpha \geq \alpha_{b,1}$, the scheme already exists; we use
\begin{align}\label{eq:etaAlpha}
\eta = \arg\max_{\eta{'}\in [\Gamma,K-1]\cap \mathbb{Z}} \{\eta{'} \ : \ \alpha_{b,\eta'}\leq \alpha\}
\end{align}
where
\begin{align}\label{eq:alphaBreak2}
\alpha_{b,\eta} = \frac{\eta-\Gamma}{\Gamma(H_K-H_\eta-1)+\eta}
\end{align}
and directly from the algorithm designed for the case of $\Gamma\geq 1$ in~\cite{ZEinterplay:16}, we get
\begin{align}
T = \max \{ 1-\gamma,  \ \frac{(K-\Gamma)(H_K-H_\eta)}{(K-\eta)+\alpha(\eta+K(H_K-H_\eta-1))}\}.
\end{align}

\section{Bounding the gap to optimal\label{sec:gapCalculation}}
This section presents the proof that the gap $\frac{T(\gamma,\alpha)}{T^*(\gamma,\alpha)}$, between the achievable $T(\gamma,\alpha)$ and the optimal $T^*(\gamma,\alpha)$, is always upper bounded by 4, which also serves as the proof of identifying the optimal $T^*(\gamma,\alpha)$ within a factor of 4.
The outer bound (lower bound) on the optimal $T^*$, is taken from \cite{ZEinterplay:16}, and it takes the form
\begin{align} \label{eq:LowerBound}
T^*(\gamma,\alpha) \geq \mathop {\text{max}}\limits_{s\in \{1, \dots, \lfloor \frac{N}{M} \rfloor \}} \frac{1}{(H_s \alpha+1-\alpha)} (H_s -\frac{Ms}{\lfloor \frac{N}{s} \rfloor}).
\end{align}

We proceed with the first case where $\alpha=0, \Gamma \leq  1$.

\subsection{Gap for $\alpha=0, \Gamma < 1$}
This and the next subsections support the gap-to-optimality statements of Theorem~\ref{thm:smallGamma} and Corollary~\ref{cor:noCSITsmallGamma} .
Our aim here is to show that $\frac{T}{T^*}< 4$, where we use the above lower bound, and where we recall that the achievable $T$ took the form
\[T = H_K-K \gamma.\]
We first see that
\begin{align}
&\frac{T}{T^*} \leq \frac{H_K-K \gamma}{\mathop {\text{max}}\limits_{s\in \{1, \dots, \lfloor \frac{N}{M} \rfloor\}} H_s -\frac{Ms}{\lfloor \frac{N}{s} \rfloor}} \leq \frac{H_K-K \gamma}{\mathop {\text{max}}\limits_{s\in \{1, \dots, \lfloor \frac{N}{M} \rfloor\}} H_s-\frac{\gamma s^2}{1- \frac{s-1}{N}}}  \label{bound31} \\
&\leq \frac{H_K-K \gamma}{\max\limits_{s\in \{1, \dots, \lfloor \frac{N}{M} \rfloor\}} H_s -\frac{\gamma s^2}{1- \frac{s-1}{K}}} \label{bound32} \\
& \leq \frac{H_K-K \gamma}{ H_{s_c} -\frac{K\gamma s_c^2}{K-s_c+1}} \label{bound33}
=:f_o(\gamma,s_c)
\end{align}
where~\eqref{bound31} holds because $\lfloor \frac{N}{s} \rfloor \leq \frac{N-(s-1)}{s}$, where~\eqref{bound32} holds because $N\geq K$, and where the last step holds because $\gamma \leq \frac{1}{K}$ and because we choose $s_c = \lfloor \sqrt{K} \rfloor$. 
We proceed to split the proof in two parts: one for $K \geq 25$, and one for $2\leq K \leq 25$.

\subsubsection{Case 1 ($\alpha = 0, K\geq 25)$} Here we see that the derivative of $f_o(\gamma,s_c)$ takes the form
\begin{align}
\frac{d f_o(\gamma,s_c) }{d \gamma} &= \frac{K}{A}(\frac{H_K s_c^2}{K-s_c+1}-H_{s_c}) \\
                                    &\geq \frac{K\log K}{A}(\frac{ (\sqrt{K}-1)^2}{K-\sqrt{K}+2}- \frac{1}{2}) \\
																		&= \frac{K\log(K)}{A}(\frac{1}{2}-\frac{\sqrt{K}+1}{K-\sqrt{K}+2}) \\
																		&\geq 0
\end{align}
where $A$ is easily seen to be positive, where the second step is because $\sqrt{K}-1 \leq s_c \leq  \sqrt{K}$ and $H_K\geq \log(K)$, and where the last step is because $0 \leq \frac{\sqrt{K}+1}{K-\sqrt{K}+2} \leq \frac{1}{2}$. Hence
\begin{align}
 \mathop {\text{max}}\limits_{\gamma \in [0,\frac{1}{K}]} f_o(\gamma,s_c) =  f_o(\gamma=\frac{1}{K},s_c) = \frac{H_K-1}{H_{s_c}-\frac{s_c^2}{K-s_c+1}}. \label{bound34}
\end{align}
Now it is easy to see that $ \frac{s_c^2}{K-s_c-1} \leq \frac{K}{K-\sqrt{K}+1}$ since $s_c = \lfloor \sqrt{K} \rfloor \leq \sqrt{K}$. Now consider the function
\[f(K) :=\frac{K}{K-\sqrt{K}+1}-\frac{\log K} {4}\]
and let us calculate its derivative
\[\frac{d f(K) }{d K} = \frac{1-\frac{\sqrt{K}}{2}}{(K-\sqrt{K}+1)^2} - \frac{1}{4K} < 0\]
which we see to be negative for any $K \geq 36$. This allows us to conclude that $ \mathop {\text{max}}\limits_{K \in [25,\infty]} f(K) = f(25) =0.3858$, and also that $\frac{s_c^2}{K-s_c-1} \leq \frac{\log K} {4} + 0.3858$.

Now let us go back to \eqref{bound34}, where using the above maximization, we can get
\begin{align}
f_o(\gamma=\frac{1}{K},s_c) & \leq  \frac{H_K-1}{H_{s_c}- (\frac{\log K} {4} + 0.3858)} \notag \\
                            & \leq \frac{H_K-1}{\frac{1}{2} \log K +\epsilon_{\infty} + \log\frac{4}{5} -(\frac{\log K} {4} + 0.3858)} \notag  \\
										        & = \frac{\log K +\epsilon_{25}-1}{\frac{1}{4} \log K +\epsilon_{\infty} + \log\frac{4}{5} - 0.3858}  \label{bound36} \\
														& \leq 4
\end{align}
where the second step is because $H_{s_c} \geq \log{s_c}+\epsilon_{\infty} \geq \log(\sqrt{K}-1)+\epsilon_{\infty} \geq \log(\frac{5}{6}\sqrt{K})+\epsilon_{\infty} = \frac{1}{2} \log K +\epsilon_{\infty} + \log\frac{5}{6} $, and where~\eqref{bound36} holds because $H_K \leq\log K +\epsilon_{25} $ since $K \geq 25$.

\subsubsection{Case 2 ($\alpha = 0, K=2,\ldots,24$)} This is an easy step, and it follows after choosing $s=1$ in the outer bound, which gives
\begin{align}
\frac{T}{T^*} \leq \frac{H_K-K\gamma}{1-\gamma} \leq \frac{K H_K}{K-1} \leq  4
\end{align}
because $\gamma \leq \frac{1}{K}$ and $K \leq 24$.

This completes the whole proof for $\alpha=0, \Gamma < 1$.

\subsection{Gap for $\alpha>0, \Gamma < 1$ \label{sec:gapAlphaBigGammaSmall}}

To bound the gap between the achievable
\begin{align}
T=\left\{ {\begin{array}{*{20}{c}}
\frac{H_K-\Gamma}{1-\alpha+\alpha H_K}, & 0 \leq \alpha < \alpha_{b,1}\\
\frac{(K-\Gamma)(H_K-H_\eta)}{(K-\eta)+\alpha(\eta+K(H_K-H_\eta-1))}, & \alpha_{b,\eta} \leq \alpha < \alpha_{b,\eta+1} \\
1-\gamma, &\frac{K-1-\Gamma}{(K-1)(1-\gamma)} \leq \alpha \leq 1
\end{array}} \right.   \label{eq:gammasmall2}
\end{align}
from Theorem~\ref{thm:smallGamma} (recall that $\eta = 1,\dots,K-2$), to the optimal $T^*$ bounded in \eqref{eq:LowerBound}, we will use the fact that
\begin{align}
\frac{(H_K-\Gamma)}{\mathop {\text{max}}\limits_{s\in \{1, \dots, \lfloor \frac{N}{M} \rfloor\}} H_s -\frac{sM}{\lfloor \frac{N}{s} \rfloor}} <4, \forall N \geq K \geq 2, \forall \Gamma < 1 \label{eq6}.
\end{align}

We will split our proof in two main cases: one for $\alpha \in [0,\alpha_{b,1})$, and another for $\alpha \in [\alpha_{b,1},1]$ (recall from \eqref{eq:alphaBreak} that $\alpha_{b,\eta} = \frac{\eta-\Gamma}{\Gamma(H_K-H_\eta-1)+\eta}$).

\subsubsection{Case 1 ($\alpha \in [0,\alpha_{b,1}]$)}
Directly from~\eqref{eq:gammasmall2}, let us use $T' :=  \frac{H_K-\Gamma}{1-\alpha+\alpha H_K}$ to denote $T$ when $\alpha \in [0,\alpha_{b,1}]$.
Now the gap is simply bounded as
\begin{align}
\frac{T'}{T^*} & \leq \frac{H_K-\Gamma}{\mathop {\text{max}}\limits_{s\in \{1, \dots,  \lfloor \frac{N}{M} \rfloor\}} \frac{1}{H_s \alpha +1-\alpha}(H_s-\frac{Ms}{\lfloor \frac{N}{s} \rfloor})(1-\alpha+\alpha H_K)} \notag \\
                   & = \frac{(H_K-\Gamma)(1-\alpha+\alpha H_s)}{\mathop {\text{max}}\limits_{s\in \{1, \dots, \lfloor \frac{N}{M} \rfloor\}}  (H_s-\frac{sM}{\lfloor \frac{N}{s} \rfloor})(1-\alpha+\alpha H_K)} <4 \notag
\end{align}
after observing that $s \leq K$ and after applying~\eqref{eq6}.

\subsubsection{Case 2 ($\alpha \in  [\alpha_{b,1},1]$)}
Let us use
\begin{align}
T'^{,\eta} := \frac{(K-\Gamma)(H_K-H_\eta)}{(K-\eta)+\alpha(\eta+K(H_K-H_\eta-1))}
\end{align}
to denote $T(\gamma,\alpha)$ in~\eqref{eq:gammasmall2} when $\alpha\in[\alpha_{b,\eta},\alpha_{b,\eta+1}], \ \eta = 1,\dots,K-2$.
For the rest of the proof, we will use the following lemma.
\begin{lemma} \label{lem:increasingInEta}
$T'^{,\eta}$ is decreasing with $\eta$, while $\alpha_{b,\eta}$ is increasing with $\eta$.
\end{lemma}
\begin{proof} See Appendix~\ref{sec:increasingInEtaSec}.
\end{proof}
Given that $T'^{,\eta}$ decreases in $\eta$, we will just prove that $\frac{T'^{,1}}{T^*} < 4$, which will automatically guarantee $\frac{T'^{,\eta}}{T^*} < 4$ for all $\eta$ and thus for all values of $\alpha$.
\vspace{3pt}

\emph{Subcase 2-a ($K\geq 25$):} From \eqref{eq6}, we see that
\begin{align}
\frac{T'^{,1}}{T^*} \! & \! \leq \frac{(K-\Gamma)(H_K-1)(1-\alpha + \alpha H_s)}{\mathop {\text{max}}\limits_{s\in \{1, \dots, K\}} (H_s- \frac{sM}{\lfloor \frac{N}{s} \rfloor})(K \!- \! 1+\alpha(1+K(H_K-2)))} \notag \\
& \leq \frac{4(K-\Gamma)(H_K-1)(1-\alpha + \alpha H_s)}{(K-1+\alpha(1+K(H_K-2)))(H_K-\Gamma)} =:\frac{4A_1}{B_1} \notag
\end{align}
where we use $A_1:=(K-\Gamma)(H_K-1)(1-\alpha + \alpha H_s)$, and where we use $B_1:=(K-1+\alpha(1+K(H_K-2)))(H_K-\Gamma)$ to denote the denominator of the last expression. To upper bound the gap by 4, we will simply show that $A_1 < B_1$. Towards this we will first show that
\begin{align}
A_1 - B_1 &= \alpha(1-H_K)(K(H_K-H_s-\Gamma)+\Gamma H_s) \notag \\
      &~~ + (\alpha-1)(K-H_K)(1-\Gamma)
\end{align}
is negative. To see this, we easily note that $(\alpha-1)(K-H_K)(1-\Gamma)\leq 0$. To guarantee that the first term above is also negative when $\Gamma \leq 1$, we just need to show (for the same $s$ that we chose before given the same parameters, but when $\alpha$ was zero) that
\begin{align} \label{eq:condition1a}
K(H_K-H_s-\Gamma)+\Gamma H_s = K(H_K-H_s-\Gamma+\gamma H_s) \geq 0 \notag
\end{align}
which is easy to see because $H_K-H_s-\Gamma+\gamma H_s = H_K-H_s+\gamma( H_s-K) \geq H_K - K\gamma \geq 0$ since $H_s-K \leq 0$ and $\gamma \in [0,1]$. This completes this part of the proof.

\emph{Subcase 2-b ($K\leq 24$):} Here we choose $s=1$ in the outer bound, and we directly have
\begin{align}
\frac{T'^{,1}}{T^*} \leq \frac{H_K-K\gamma}{1-\gamma} \leq \frac{K H_K}{K-1} \leq  4
\end{align}
because $\gamma \leq \frac{1}{K}$ and $K \leq 24$.

This completes the proof for $\alpha>0, \Gamma \leq 1$, and also completes the entire proof.

\section{Conclusions \label{sec:conclusions}}

Our work considered the wireless MISO BC in the presence of two powerful but scarce resources: feedback to the transmitters, and caches at the receivers. Motivated by realistic expectations that cache sizes --- at wireless receivers/end-users --- might be small (\!\!\!\cite{EJR:15}), and motivated by the well known limitations in getting good-quality and timely feedback, the work combines these scarce and highly connected resources, to conclude that
we can attribute non-trivial performance gains (or non-trivial CSI reductions) even if the caches correspond to vanishingly small $\gamma\rightarrow 0$.
This synergy between feedback and caching, allows for a serious consideration of scenarios where even microscopic fractions of the library can be placed at different caches across the network, better facilitating the coexistence of modestly-sized caches and large libraries.

\section{Appendix\label{sec:additionalProofs}}

\subsection{Proof of Corollary~\ref{cor:asymptotic2} \label{sec:asymptoticProofSmallGamma}}
Our aim here is to show that for large $K$, and when $\gamma\ll 1$, the achievable $T$ (both from Theorem~\ref{thm:smallGamma} corresponding to the case of $\Gamma<1$, but also from \eqref{eq:GammaLarge}), has a gap to optimal that is at most 2, for all $\alpha$.
We first consider the scenario where $\alpha = 0$, and note that 
\begin{align}
T(\Gamma \geq 1,\alpha = 0) & = H_K-H_{K\gamma} \leq \log(K) + \epsilon_2 - \log(K\gamma) \notag \\
  &\leq \log (\frac{1}{\gamma}) +  \epsilon_2  \notag \\
T(\Gamma <1,\alpha = 0)&= H_K-K \gamma \leq \log(K) + \epsilon_2 \notag \\ & \leq \log (\frac{1}{\gamma}) +  \epsilon_2\notag
\end{align}
and thus note that in both cases, we have that
\[T \leq \log (\frac{1}{\gamma}) +  \epsilon_2, \forall \Gamma \geq 0\] which means that
\begin{align}
 \frac{T}{T^*} \leq  \frac{\log (\frac{1}{\gamma}) +  \epsilon_2} {\mathop {\text{max}}\limits_{s\in \{1, \dots,  \lfloor \frac{N}{M} \rfloor \}} H_s -\frac{s M}{\lfloor \frac{N}{s} \rfloor}}  \leq \frac{\log (\frac{1}{\gamma}) +  \epsilon_2} { H_{s_c} -\frac{M s_c}{\lfloor \frac{N}{s_c} \rfloor}} \label{eq:asymptoticoptimal}
\end{align}
for any $s_c \in \{1, \dots,  \lfloor \frac{N}{M} \rfloor  \}$.
Now let us choose $s_c= \lfloor \sqrt{\frac{1}{\gamma}} \rfloor$, and note that $N \geq M s_c^2 \gg s_c$, which means that $\frac{\lfloor \frac{N}{s_c} \rfloor}{\frac{N}{s_c}} \rightarrow 1$. Consequently, from \eqref{eq:asymptoticoptimal}, for both cases, we have
\begin{align}
\lim_{K \rightarrow \infty} \frac{T}{T^*} &\leq \lim_{K \rightarrow \infty} \frac{\log (\frac{1}{\gamma}) +  \epsilon_2}{\log (s_c) - \gamma s_c^2} \notag \\
 &= \lim_{K \rightarrow \infty} \frac{2 \log (s_c)+ \epsilon_2 }{\log (s_c) - 1} = 2
\end{align}
proving the tighter gap to optimal, which is at most 2, for the case of $\alpha = 0$.

Then, directly from Section~\ref{sec:gapAlphaBigGammaSmall} for the case of $\Gamma< 1$ case, and from \cite{ZEinterplay:16} for the case of $\Gamma\geq 1$, we know that the gap decreases when $\alpha>0$, which concludes the proof.

\subsection{Proof of Lemma~\ref{lem:increasingInEta}\label{sec:increasingInEtaSec}}
We need to show that $\alpha_{b,\eta}$ is increasing in $\eta$, and that $T'^{,\eta}$ is decreasing in $\eta$.
The first follows by noting that
\begin{align}
& \alpha_{b,\eta+1} \!-\! \alpha_{b,\eta} \! = \! \frac{(H_K-H_\eta)}{(H_K\!-\! H_\eta\!-\! 1)\!+\! \frac{\eta}{\Gamma}} \! - \! \frac{(H_K-H_{\eta+1})}{(H_K\!-\! H_{\eta+1}\!-\! 1)\!+\! \frac{\eta+1}{\Gamma}} \notag \\
               & = \frac{H_K-H_\eta+\frac{\eta-\Gamma}{\eta+1}}{((H_K\!-\! H_\eta\!-\! 1)\!+\! \frac{\eta}{\Gamma})((H_K\!-\! H_{\eta+1}\!-\! 1)\!+\! \frac{\eta+1}{\Gamma})} > 0
\end{align}
which holds because $\eta \geq \Gamma$.

To see that $T'^{,\eta}$ decreases in $\eta$, after simplifying notation by letting
\begin{align}
D_\eta &:= (K-\eta)+\alpha(\eta+K(H_K-H_\eta-1)) \notag \\
D_{\eta+1} &= (K-\eta-1)+\alpha(\eta+1+K(H_K-H_{\eta+1}-1)) \notag
\end{align}
to denote the denominators of $T'^{,\eta}$ and of $T'^{,\eta+1}$ respectively, we see that
\begin{align}
 & \frac{T'^{,\eta}-T'^{,\eta+1}}{K-\Gamma}= \frac{H_K-H_\eta}{D_\eta}-\frac{H_K-H_{\eta+1}}{D_{\eta+1}}  \notag \\
		& = \frac{(\frac{K-\eta}{\eta+1}+H_\eta-H_K)(1-\alpha)}{D_\eta D_{\eta+1}} \notag \\
		& = \frac{(\frac{K-\eta}{\eta+1}-(\frac{1}{\eta+1}+\frac{1}{\eta+2}+\cdots+\frac{1}{K}))(1-\alpha)}{D_\eta D_{\eta+1}} > 0
\end{align}
which holds because $\eta \leq K$ and $\frac{1}{\eta+1} \geq \frac{1}{\eta+i}, \forall i \in [1, K-\eta]\cap \Z$. The above inequality is strict when $\eta>\Gamma$. This completes the proof.


\bibliographystyle{IEEEtran}
\bibliography{IEEEabrv,final_refs}

\end{document}